\documentclass[11pt,letterpaper]{amsart}
\usepackage{graphicx}    
\usepackage{bbm, xcolor}

\usepackage{subfigure}

\usepackage{epsfig,amsmath,euscript,amssymb}
\usepackage{amsthm}
\newtheorem{lemma}{Lemma}
\newtheorem{theorem}{Theorem}
\newtheorem{proposition}{Proposition}
\newtheorem{remark}{Remark}

\begin{document}

\title[Quasi-Gaussian HJM model]
{Explosion in the quasi-Gaussian HJM model}

\author{Dan Pirjol}
\email{
dpirjol@gmail.com}

\author{Lingjiong Zhu}
\email{
ling@cims.nyu.edu}

\date{}


\begin{abstract}
We study the explosion of the solutions of the SDE in the 
quasi-Gaussian HJM model with a CEV-type volatility. The quasi-Gaussian HJM 
models are a popular approach for modeling the dynamics of the yield curve. 
This is due to their low dimensional Markovian representation which simplifies 
their numerical implementation and simulation. We show rigorously that the 
short rate in these models explodes in finite time with positive 
probability, under certain assumptions for the model parameters, and that the
explosion occurs in finite time with probability one under some stronger 
assumptions. We discuss the implications of these results for the pricing
of the zero coupon bonds and Eurodollar futures under this model.
\end{abstract}

\maketitle

\section{Introduction}

The quasi-Gaussian Heath-Jarrow-Morton (HJM) models \cite{APbook,review,B,C,RS} are frequently used
in financial practice for modeling the dynamics of the yield curve \cite{APbook}.
They were introduced as a simpler alternative to the HJM model \cite{HJM}, which
describe the dynamics of the yield curve $f(t,T)$ as the stochastic
differential equation
\begin{equation}
df(t,T) = \sigma_f(t,T) dW(t) + 
\sigma_f(t,T) \left( \int_t^T \sigma_f(t,s) ds \right) dt \,,
\end{equation}
where $W(t)$ is a vector Brownian motion under the risk-neutral measure 
$\mathbb{Q}$, and $(\sigma_f(t,T))_{t\leq T}$ is a family of vector processes.
The dynamical variable in the HJM models is the forward rate $f(t,T)$ for
maturity $T$. 

The quasi-Gaussian HJM models assume a separable form for the volatility 
function $\sigma_f(t,T) = g(T) h_t$ where $g$ is a deterministic vector 
function and $(h_t)$ is a $k\times k$ matrix stochastic process. Such models 
admit a Markov 
representation of the dynamics of the yield curve in terms of $k+\frac12 k(k+1)$ 
state variables. This simplification aids considerably with the simulation of
these models, which can be performed using Monte Carlo or finite difference
methods \cite{CaZhu,PhD}.

We consider in this paper the one-factor quasi-Gaussian HJM model with
volatility specification $\sigma_f(t,T) = k(t,T) \sigma_r(r_t)$ where 
$k(t,T)=e^{-\beta(T-t)}$, and $\sigma(r_t)$ is the volatility of the short rate
$r_t = f(t,t)$. This model admits a two-state Markov representation.

It has been noted in \cite{Morton,HJM} that in HJM models with log-normal
volatility specification, that is for which $\sigma_f(t,T) = \sigma
f(t,T)$, the rates explode to infinity with probability one, and 
zero coupon bond prices approach zero. 
See also \cite{Wissel} for a general study of the conditions
for the existence of strong solutions to stochastic differential equations (SDEs) of HJM type.
A similar explosion appears in a two-dimensional model studied in \cite{Hogan}.
It is natural to ask if such explosions are present also in the
quasi-Gaussian HJM models. Models of this type with parametric
volatility are used in financial practice
for modeling the swaption volatility skew \cite{Chibane,quadratic,CaZhu}. 
Non-parametric forms have been also considered recently in the literature
\cite{GJQ,ChDup}.

We recall that singular behavior is also observed for certain derivatives
prices in short rate
log-normal interest rates models \cite{AP,APbook}. It was observed by 
Hogan and Weintraub \cite{HW} that Eurodollar futures prices are infinite 
in the Dothan and Black-Karasinski models. 
A milder singularity is also present in finite tenor log-normal models,
such as the Black-Derman-Toy model, manifested as a rapid increase of the 
Eurodollar futures convexity adjustment as the volatility increases above
a threshold value \cite{QF}.
This singularity can be avoided
by formulating the models by specifying the distributional properties
of rates with finite tenor \cite{SS}. 
This line of argument led to the formulation of the LIBOR Market Models
which are free of singularities \cite{APbook}.

In a recent work \cite{ORL}, we studied the small-noise limit of the
log-normal quasi-Gaussian model,
using a deterministic approximation, and showed rigorously that the short 
rate may explode to infinity in a finite time.
More precisely, it was shown in \cite{ORL} that for sufficiently small 
mean-reversion parameter $\beta$, the small-noise approximation for the short 
rate $r_t$ has an explosion in 
finite time, and an upper bound is given on the explosion time, which
is saturated in the flat forward rate limit.

In this paper, we extend these results in two directions:

(i) We consider a wider class of quasi-Gaussian HJM models with 
a constant elasticity of variance (CEV)-type volatility specification. 
This includes the log-normal model as a special case.
We also consider the case of the displaced log-normal model. 
These volatility specifications are widely used 
by practitioners \cite{Chibane,APbook}.

(ii) The Brownian noise is taken into account.
This requires the study of the explosion of the solutions of a 
two-dimensional stochastic differential equation. 
Mathematically, it is well known that for one-dimensional diffusion
processes there is the celebrated Feller criterion \cite{Feller,McKean} 
for explosion/non-explosion, see e.g. \cite{KaratzasShreve,RossPinsky} for 
overviews. 
This is a sufficient and necessary condition under which there is an explosion
in finite time. We note that the distribution of the explosion time has
been also studied recently \cite{KarRuf}.

For $d$-dimensional stochastic differential equations with $d>1$,
to the best of our knowledge, there is no sufficient and necessary condition
for explosion. 
Several sufficient conditions for explosions have been presented in 
the literature for multi-dimensional diffusions \cite{StroockVaradhan}.
The Khasminskii criterion for explosion is presented in \cite{McKean,RossPinsky}.
The method of the Lyapunov function was presented in \cite{CK,Khasminskii}.
This was extended to a non-linear Lyapunov method by \cite{exp}.
The application of these conditions is non-trivial, and checking that 
the conditions required hold is sometimes very challenging.

We rely on the sufficient conditions for explosion with positive probability 
and explosion with probability one given in \cite{CK,Khasminskii}. 
The main tool is the construction of some delicate Lyapunov functions that 
satisfy certain non-trivial conditions \cite{CK,Khasminskii}. 

We show rigorously that under certain conditions, 
in the CEV-type model with exponent in a certain range $(\frac12,1]$,
including the log-normal case, the short rate explodes in finite time with 
positive probability.
We also show rigorously that under additional
assumptions, the explosion occurs with probability one. 

The explosion phenomenon that we prove rigorously has implications 
for the practical use of the model for pricing and simulation. Such
explosions are observed in practical applications of the model, and we
illustrate them on a numerical example in the log-normal quasi-Gaussian
HJM model in Section~\ref{sec:example}. This phenomenon implies 
the collapse of the zero coupon bond 
prices, similar to that occurring in  the log-normal HJM model \cite{HJM},
and an explosion of interest rate derivatives linked to the LIBOR rate, 
in particular the Eurodollar futures prices.
This introduces a limitation in the applicability of the 
model for pricing these products to maturities smaller than the explosion time. 

The paper is organized as follows. In Section~\ref{sec:model},
we introduce the model, and discuss its use in the literature. 
In Section~\ref{sec:main}, we present rigorous
results giving sufficient conditions for explosion in finite time with
positive probability in the quasi-Gaussian HJM model
with CEV-like volatility specification.
Furthermore, under stronger assumptions,
we can show that the explosion occurs in some finite time with probability one.
In Section~\ref{sec:imply}, we discussion the implications of our results
to the pricing of the zero coupon bond and the Eurodollar futures.
Finally, the proofs are collected in the Appendix.

\section{One factor quasi-Gaussian HJM model}\label{sec:model}

We will consider in this paper a class of one-factor quasi-Gaussian 
HJM models, defined by the 
volatility specification
\begin{equation}\label{specmodel}
\sigma_f(t,T) = \sigma_r( r_t ) e^{-\beta (T-t)} \,.
\end{equation}
Several parametric choices for the short rate volatility function
$\sigma_r(x)$ have been considered in the literature, including:

\begin{itemize}
\item[(i)] Log-normal model \cite{CaZhu}: $\sigma_r(x) = \sigma x$;

\item[(ii)] Displaced log-normal model, also known in the literature as the
linear Cheyette model \cite{Chibane,APbook}:
$\sigma_r(x) = \sigma (x + a)\,$;

\item[(iii)] CEV-type model \cite{CaZhu}:
$\sigma_r(x) = \sigma x^\gamma$, where $\gamma\in(0,1]\,$.
\end{itemize}

The simulation of the model with the volatility specification \eqref{specmodel} 
can be reduced to simulating the stochastic
differential equation for the two variables $\{ x_t,y_t\}_{t\geq 0}$ 
\cite{RS,APbook}
\begin{eqnarray}\label{xySDE}
&& dx_t = (y_t - \beta x_t) dt + \sigma_r(\lambda(t) + x_t) dW_t  ,\\
&& dy_t = (\sigma_r^2(\lambda(t) + x_t) - 2\beta y_t) dt ,\nonumber
\end{eqnarray}
with initial condition $x_0 = y_0 = 0$. Here $\lambda(t) = f(0,t)$ is
the forward short rate, giving the initial yield curve. 
The price of the zero coupon bond with maturity $T$ is
\begin{equation}\label{ZCB}
P(t,T) = \frac{P(0,T)}{P(0,t)}
\exp\left(- G(t,T) x_t - \frac12 G^2(t,T) y_t \right)\,,
\end{equation}
with $G(t,T) = \frac{1}{\beta}(1-e^{-\beta (T-t)})$ a non-negative 
deterministic function \cite{APbook}. 
The short rate is $r_t := f(t,t) = \lambda(t) + x_t$.

Under the CEV-type volatility
$\sigma_{r}(x)=\sigma x^\gamma$,
the equations (\ref{xySDE}) can be expressed in terms of the short rate as
\begin{eqnarray}\label{rySDE}
&& dr_t = \sigma r_{t}^\gamma dW_t + 
          (y_t - \beta r_t + \beta \lambda(t) + \lambda'(t)) dt \,,\\
&& dy_t = (\sigma^2 r_t^{2\gamma} - 2\beta y_t)  dt \,,\nonumber
\end{eqnarray}
with the initial condition $r_0 = \lambda_0 := \lambda(0) > 0$ and $y_0 = 0$.

One potential complication with 
the usual CEV volatility specification $\sigma_{r}(x)=\sigma x^{\gamma}$ is 
related to the non-uniqueness of the solution of the SDE (\ref{rySDE}) for $0 < \gamma < 1$ \cite{Feller}.  
Recall that for 
the usual CEV model \cite{CEV}, given by the SDE with $0<\gamma \leq 1$,
\begin{equation}
dx_t = \sigma x_t^\gamma dW_t + \mu x_t dt\,,
\end{equation}
the origin $x=0$ is a regular boundary for $0<\gamma<\frac12$, and
an exit boundary for $\frac12 \leq \gamma < 1$. For the geometric Brownian 
motion case $\gamma=1$,
the point zero is a natural boundary. For $0<\gamma<\frac12$, the solution
of the SDE is  not unique, and an additional boundary condition must be 
imposed at $x=0$ in order to ensure uniqueness. 

In order to avoid singular behavior near the origin $r=0$,
practitioners use various modifications of the
quasi-Gaussian model with CEV volatility specification near the $r=0$ point,
see e.g. Section 4.3 of \cite{CheyettePhD}. This work describes three possible
modifications: (a) $\sigma_r(x) \to \sigma |x|^\gamma$; (b)
$\sigma_r(x) \to 0$ for $|x|\leq \varepsilon$; (c)
$\sigma_r(x) \to \sigma \varepsilon^\gamma$ for $|x|\leq \varepsilon$
with $\varepsilon>0$ a small cutoff.

In this paper, following
\cite{AA}, we consider the modified quasi-Gaussian HJM model with a CEV-type 
volatility specification
\begin{equation}\label{epsCEV}
\sigma_r(x) = \sigma x \min(x^{\gamma-1} , \varepsilon^{\gamma-1})\,,
\end{equation}
with $\varepsilon>0$ small and $0<\gamma\leq 1$. 
We call this the $\varepsilon-$CEV quasi-Gaussian HJM model. 

Note that as $\varepsilon=0$, this reduces to the usual
CEV volatility specification $\sigma_{r}(x)=\sigma x^{\gamma}$, see e.g. \cite{CaZhu}.
The modification $\varepsilon>0$ impacts only the
region of small $0 < r_{t} < \varepsilon$, where the process is identical with
the log-normal model with the volatility $\sigma\varepsilon^{\gamma-1}$, and leaves unchanged the behavior of 
the process for large $r_{t}$, which is relevant for the study of the explosions of $r_{t}$.
The modification is only required for $0< \gamma < 1$.
When $\gamma=1$, the equation \eqref{epsCEV} reduces to $\sigma_{r}(x)=\sigma x$, 
which coincides with the log-normal model.

With the volatility specification \eqref{epsCEV}, we will study the 2-dimensional SDE with $\varepsilon>0$
\begin{align}\label{rSDE}
&dr_{t}=(y_{t}-\beta r_{t}+\beta\lambda(t)+\lambda'(t))dt
+\sigma r_t \min(r_t^{\gamma-1},\varepsilon^{\gamma-1})dW_{t},
\\
\label{ySDE}
&dy_{t}=(\sigma^{2} r_{t}^{2} \min(r_t^{2\gamma-2}, \varepsilon^{2\gamma-2})
-2\beta y_{t})dt,
\end{align}
with the initial condition $r_{0}=\lambda(0)> \varepsilon$ and $y_{0}=0$.

In the special case $\gamma=1$, \eqref{rSDE},\eqref{ySDE} reduces to 
the log-normal model
\begin{eqnarray}
&& dr_t = \sigma r_{t}dW_t + 
          (y_t - \beta r_t + \beta \lambda(t) + \lambda'(t)) dt \,,\label{rLN}\\
&& dy_t = (\sigma^2 r_t^{2} - 2\beta y_t)  dt \,,\label{yLN}
\end{eqnarray}
with the initial condition $r_0 = \lambda_0 := \lambda(0) > 0$ and $y_0 = 0$.

Assume that $\lambda'(t) + \beta \lambda(t) \geq 0$ and $r_0>0$. 
Then the solutions of \eqref{rLN}, \eqref{yLN} are positive with probability one
\begin{equation}\label{rpositive}
\mathbb{P}(r_t > 0)=1,\qquad\text{for all $t\geq 0$}.
\end{equation} 
The result follows by noting that 
\begin{equation}\label{ypositive}
y_t = \sigma^2 \int_0^t r_s^2 e^{-2\beta(t-s)}
ds >0,
\end{equation}
almost surely for every $t>0$, and then follows by an application of
the comparison theorem (Theorem 1.1 in \cite{Yamada} and Theorem 5.2.18 in 
\cite{KaratzasShreve}). See also the Appendix D in \cite{NAK} for
a proof of this result. 

This implies that the origin $r=0$ is a natural boundary for this diffusion. 
For the time-homogeneous case $\lambda(t)=\lambda_0$ we use a similar argument
to prove the same result for the $\varepsilon-$CEV model 
with general $\gamma \in (\frac12,1]$, see the argument around Eq.~(\ref{43}).

The SDE for the displaced log-normal model
$\sigma_r(x)=\sigma (x+a)$ reduces to that for the log-normal case by the 
substitutions $r_t + a \to r_t, \lambda(t) + a \to \lambda(t)$. Expressed
in terms of $r_t,y_t$, this is
\begin{eqnarray}
&& dr_t = \sigma (r_t + a) dW_t + (y_t - \beta r_t + \beta \lambda(t) + \lambda'(t))dt\,, \\
&& dy_t = (\sigma^2 (r_t + a)^2 - 2\beta y_t) dt\,,
\end{eqnarray}
with initial conditions $r_0=\lambda(0), y_0=0$. 
Defining $\tilde r_t = r_t + a$ the shifted short rate, we have
\begin{eqnarray}
&& d\tilde r_t = \sigma \tilde r_t dW_t + 
 (y_t - \beta \tilde r_t + \beta (\lambda(t) +a) + \lambda'(t))dt\,, \\
&& dy_t = (\sigma^2 \tilde r_t^2 - 2\beta y_t) dt\,,
\end{eqnarray}
started at $\tilde r_t=\lambda(0)+a, y_0=0$. Redefining $\lambda(t)+a\to 
\lambda(t)$, the shift parameter $a$ disappears, and the resulting SDE is 
identical to that for the log-normal $\gamma=1$ model. 

Under this model negative values for $r_t$ can also be accommodated,
with a floor on the short rate $r_t > -a$.
We will assume that $r_0+ a > 0$, and then $r_t > - a$ for any $t>0$.
All the results for $\gamma=1$ apply also to the displaced log-normal model
with minimal substitutions.

In \cite{ORL}, we studied the small-noise deterministic limit of the SDE
\eqref{rSDE},\eqref{ySDE} in the log-normal case $\gamma=1$
\begin{eqnarray}\label{ODE}
&& r'(t) = y(t) - \beta r(t) + \beta \lambda(t) + \lambda'(t),\\
&& y'(t) = \sigma^2 (r(t))^2 - 2\beta y(t) \,,\nonumber
\end{eqnarray}
with $r(0)=\lambda(0)=\lambda_{0}>0$ and $y(0)=0$.
In the small-noise limit, it is proved rigorously in \cite{ORL} that
for sufficiently large $\beta$ or sufficiently small $\sigma$, 
the short rate $r(t)$ is uniformly bounded, and hence there is no explosion.
When $\beta=0$, the short rate explodes in a finite time, and an upper bound 
is given for the explosion time (Proposition 4 in \cite{ORL}). 
Under the further assumption that $\lambda(t)\equiv\lambda_{0}$,
the upper bound for the explosion time given in Proposition 4 of \cite{ORL} 
is sharp. The $\beta>0$ case is also considered, under the simpler setting
of a time homogeneous model $\lambda(t)\equiv\lambda_{0}$.
For this case it is shown in \cite{ORL} that when 
$\beta<\beta_{C}:=\sigma\sqrt{2\lambda_{0}}$, 
the explosion occurs at a finite time and
when $\beta\geq\beta_{C}$, we have
$\lim_{t\rightarrow\infty}r(t)=\frac{\beta^{2}}{\sigma^{2}}
(1-\sqrt{1-\frac{2\sigma^{2}\lambda_{0}}{\beta^{2}}})$,
and there is no explosion.

In this paper we would like to study directly the original stochastic system 
\eqref{rSDE},\eqref{ySDE} in the presence of random noise. 
We will show rigorously that the solutions of the stochastic system 
\eqref{rSDE}, \eqref{ySDE} may explode with 
non-zero probability for $\gamma \in (\frac12,1]$, 
and under some additional assumptions with probability one.

\subsection{Numerical example for $\gamma=1$}
\label{sec:example}

Such explosions are indeed observed in numerical simulations of the
stochastic system \eqref{rSDE},\eqref{ySDE}.
We illustrate this phenomenon for the log-normal model $\gamma=1$
in Figure~\ref{fig:beta}, 
which shows sample paths for $\{r_t \}_{t\geq 0}$ for
several choices of the model parameters $\sigma$, $\lambda(:=\lambda_0),\beta$.
These results were obtained by numerical simulation of the stochastic
differential equations \eqref{rSDE},\eqref{ySDE} by Euler discretization with time step
$\tau=0.01$.

In Figure~\ref{fig:beta}, we fix $\sigma=0.2,\lambda=0.1$ and 
consider two values of $\beta$. The left plot shows sample paths for $r_t$ with
$\beta=0$. The paths explode at various times, which is expected in the
presence of the Brownian noise. In the small-noise limit studied in
\cite{ORL}, the explosion time is deterministic. The corresponding explosion
time can be found in closed form for $\lambda(t) = \lambda_0$, and is
given in Proposition 4 of \cite{ORL}. The prediction is shown in 
Figure~\ref{fig:beta} (left) as the red vertical line. 

The right plot in Figure~\ref{fig:beta} shows sample paths for $\beta=0.05$.
There is still explosion, but the explosion tends to occur at longer maturities.
This is in qualitative agreement with the behavior expected in the 
small-noise limit \cite{ORL} where it
was shown that increasing $\beta$ delays the explosion time, and suppresses
it completely for $\beta \geq \beta_C=\sigma\sqrt{2\lambda_0}$. 
For the parameters considered
in Figure~\ref{fig:beta} the small-noise critical value is $\beta_C=0.089$.
In the stochastic case, taking $\beta=0.1$ (not shown) the explosion is 
further delayed to longer maturities, or completely suppressed.

\begin{figure}
\centering
\begin{subfigure}
  \centering
  \includegraphics[width=2.8in]{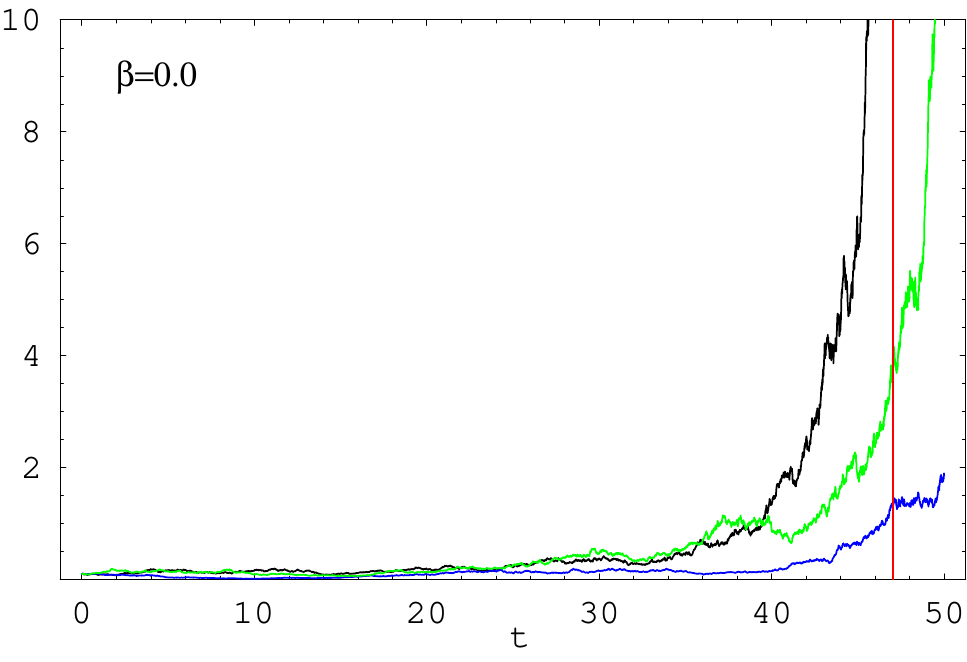}
\end{subfigure}
\begin{subfigure}
  \centering
  \includegraphics[width=2.8in]{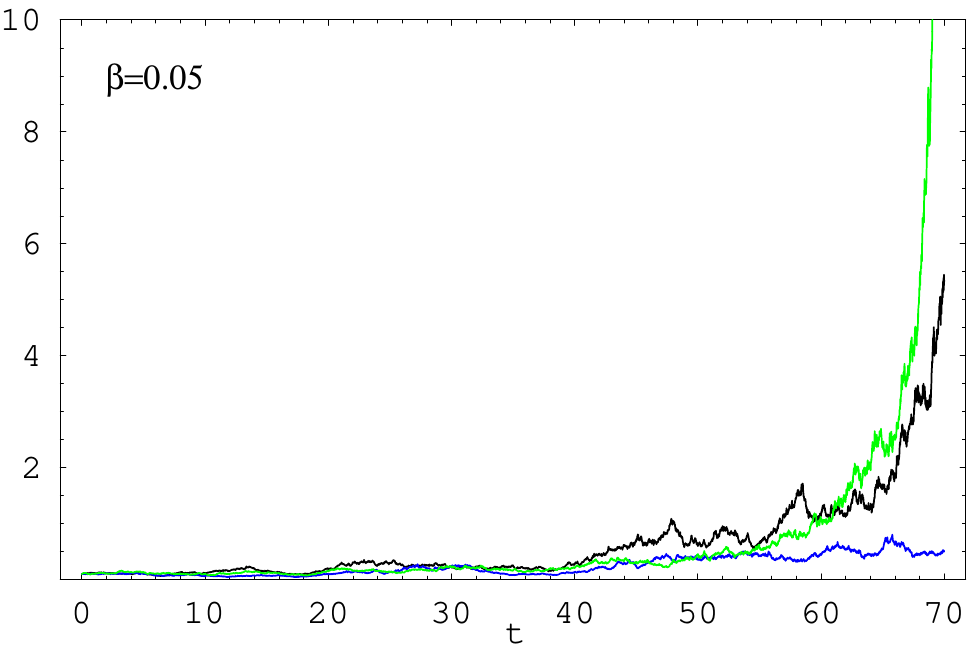}
\end{subfigure}
\caption{Sample paths for $\{ r_t\}_{t\geq 0}$ for 
$\sigma=0.2$ and $\lambda=0.1$ in the log-normal quasi-Gaussian HJM model.
Left plot: sample paths with $\beta=0$. The red vertical
line is at $t_{\rm exp}=47.03$y, which is the small-noise explosion
time following from Proposition 4 in \cite{ORL}.
Right plot: $\beta=0.05$.}
\label{fig:beta}
\end{figure}

\section{Explosion of the CEV-type quasi-Gaussian HJM model}
\label{sec:main}

Assume that the forward rate $\lambda(t)$ satisfies the inequality,
\begin{equation}\label{assump1}
\lambda'(t) + \beta \lambda(t) \geq \beta \lambda(0)\,.
\end{equation}
By a comparison argument, the solutions of \eqref{rSDE},\eqref{ySDE} are bounded from below
by the solutions of the time-homogeneous SDE obtained by replacing
$\lambda(t) \to \lambda_{0}=\lambda(0)$. Thus, for the purpose of studying
the explosions of the solutions of the SDE \eqref{rSDE},\eqref{ySDE} it is sufficient
to study the corresponding time-homogeneous SDE with constant $\lambda(t)=
\lambda_0$
\begin{align}\label{rSDE2}
&dr_{t}=(y_{t}-\beta r_{t}+\beta\lambda_0)dt
+\sigma r_t \min(r_t^{\gamma-1},\varepsilon^{\gamma-1})dW_{t},
\\
\label{ySDE2}
&dy_{t}=(\sigma^{2} r_{t}^{2} \min(r_t^{2\gamma-2}, \varepsilon^{2\gamma-2})
-2\beta y_{t})dt,
\end{align}
with the initial condition $r_{0}=\lambda_0> \varepsilon$ and $y_{0}=0$.

The coefficients of this SDE satisfy a local Lipschitz condition. 
For $0< \gamma \leq \frac12$ they also satisfy a sublinear growth
condition and global Lipschitz condition.
Thus we can apply the standard result, see for example Theorem 5.2.9 in
\cite{KaratzasShreve}, to conclude that the SDE has a unique
strong solution, which is furthermore square integrable and thus non-explosive. 
On the other hand we show that for $\frac12 < \gamma \leq 1$ the solution can
explode to infinity in finite time with non-zero probability.


The infinitesimal generator of this diffusion is
\begin{eqnarray}\label{Leps}
\mathcal{L}_\varepsilon V(r,y) &=&
(\sigma^2 r^2 \min(r^{2\gamma-2}, \varepsilon^{2\gamma-2}) - 2\beta y)\partial_y V\\
& & + (y- \beta r + \beta r_0) \partial_r V 
+ \frac12\sigma^2 r^2 \min(  r^{2\gamma-2}, \varepsilon^{2\gamma-2})
\partial_r^2 V\,. 
\nonumber
\end{eqnarray}

We would like to study the explosion time of this diffusion, defined as
\begin{equation}
\tau:=\sup\{t>0:y_{t}<\infty,r_{t}<\infty\}.
\end{equation}

\begin{figure}[t!]
\begin{center}
\includegraphics[height=65mm]{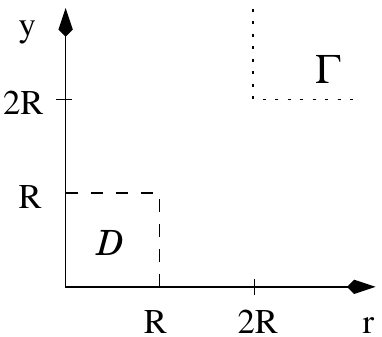}
\end{center}
\caption{
Regions $\mathcal{D}$ and $\Gamma$ for the application of 
Proposition~\ref{Thm1}.}
\label{Fig:regD}
\end{figure}

We present a few preliminary results which will be used in our proof.
The following theorem was proved in \cite{CK}, see Theorem 1.

\begin{proposition}[Theorem 1, \cite{CK}]\label{Thm1}
Let $\mathcal{D}\subset\mathbb{R}^{d}$ be a bounded open set with regular 
boundary $\partial\mathcal{D}$ and let $\mathcal{D}^{c}$ be the complement 
of $\mathcal{D}$. Consider the $d$-dimensional diffusion 
$dX_t = \sigma(X_t) dW_t + b(X_t)dt$ where the coefficients $\sigma(\cdot),
b(\cdot)$ are Lipschitz continuous on any compact subset of 
$\mathbb{R}^{d}$ for any $t\geq t_{0}$.
Moreover, there exists a positive function $V(t,x)\in C^{1,2}([t_{0},\infty)\times\mathcal{D}^{c})$
and positive constants $K_{1}$, $K_{2}<K_{3}$ and $C$ such that

(A.1) $\sup_{t\geq t_{0},x\in\mathcal{D}^{c}}V(t,x)=K_{1}<\infty$.

(A.2) $\sup_{t\geq t_{0},x\in\partial\mathcal{D}}V(t,x)=K_{2}<\inf_{t\geq t_{0},x\in\Gamma}V(t,x)=K_{3}$,
for some set $\Gamma\subset\mathcal{D}^{c}$.

(A.3) $\mathcal{L}V(t,x)\geq CV(t,x)$ for every $t\geq t_{0}$, 
$x\in\mathcal{D}^{c}$, 
where $\mathcal{L}$ is the infinitesimal generator of $X_t$.

Then, the explosion eventually occurs with positive probability if the process
starts at a future time $t_{1}\geq t_{0}$ at a point $x\in\Gamma$.
\end{proposition}

\begin{proposition}[Theorem 2 in \cite{CK}]\label{Thm2}
Assume the conditions in Proposition~\ref{Thm1} are satisfied. 
Then, we have the almost sure explosion provided the additional assumptions hold:

(A.4) $\inf_{t\geq t_{0},x\in\mathcal{D}^{c}}V(t,x)=K_{0}>0$;

(A.5) For any $t\geq t_{0}$, $x\in\partial\mathcal{D}$,
$\mathbb{P}^{t,x}(\tau_{\Gamma}<\infty)=1$,
where $\tau_{\Gamma}$ is the first hitting time of the set $\Gamma$.
\end{proposition}

Theorem 1 in \cite{CK} is a generalization of Theorem 3.6 in \cite{Khasminskii}.
Intuitively, it relates the explosion of the solution of a stochastic
differential equation to the behavior of an appropriately defined Lyapunov
function for large values in the space domain. 

Unlike the case of one-dimensional diffusion processes,  where a sufficient 
and necessary condition for explosion is given by Feller's criterion \cite{McKean}, 
for multidimensional diffusions, there are many different theoretical results 
giving sufficient conditions for explosions \cite{McKean,CK,Khasminskii}. 
For our purpose, Theorem 1 in \cite{CK} suffices. The
strategy of the proof will be to construct
an appropriate Lyapunov function, and show that for the two-dimensional 
SDE model \eqref{rSDE2}, \eqref{ySDE2}, explosion occurs with positive 
probability. The main result of this paper follows.


\begin{theorem}\label{mainThm}
Assume $\lambda(t)\equiv r_{0}>\varepsilon > 0$ and $\beta \geq 0$.

(a) For $\gamma \in (0,\frac12]$ the solution of the 
SDE \eqref{rSDE2},\eqref{ySDE2} $\{r_{t},y_{t}\}_{t\geq 0}$ is non-explosive.

(b) For $\gamma \in (\frac12,1]$ the solution of the SDE 
\eqref{rSDE2},\eqref{ySDE2} $\{r_{t},y_{t}\}_{t\geq 0}$  explodes with non-zero
probability $P(\tau<\infty)>0$ provided that any one of the following conditions
is satisfied for at least one set of $(\delta_1,\delta_2)$, where
$\delta_{1},\delta_{2} >0$ are positive constants
satisfying $(1+\delta_1)(1+\delta_2) = 2\gamma$.

(i) $\sup_{R\geq \varepsilon}F(R;\beta,\sigma) >0$ where the function 
$F(R;\beta,\sigma)$ is defined by
\begin{equation}
F(R;\beta,\sigma):= R^{2\gamma} -
\left(2\beta+\frac12\sigma^2 \delta_2(\delta_2+1)\right)
\left(
\frac{1}{\delta_1 \sigma^2}
(1+R)^{\delta_1+1} +
\frac{1}{\delta_2}  
R^{2\gamma-1} (1+R)^{\delta_2+1} \right)\,.
\end{equation}

(ii) $\sup_{R\geq \varepsilon}(G(R) - (2\beta+\frac12 \sigma^2\delta_2(\delta_2+1)))
\geq 0$ where the function $G(R)$ is defined by
\begin{equation}\label{Gdef}
G(R):= \frac{\delta_2 R}{(1+R)^{\delta_2+1}}\,.
\end{equation}
\end{theorem}


\begin{remark}\label{remarkCondition}
Under the assumptions in Theorem \ref{mainThm}, 
$V$, $K_{1}$, $K_{2}$, $K_{3}$ and $C$  from the conditions in 
Proposition \ref{Thm1} are given by
\begin{align}
&V(r,y) =C_{1}-\frac{C_{2}}{(1+y)^{\delta_1}}-\frac{C_{3}}{(1+r)^{\delta_2}},
\\
&K_{1}=C_{1},
\\
&K_{2}=\sup_{(r,y)\in\partial\mathcal{D}}
V(r,y)=C_{1}-\frac{C_{2}}{(1+R)^{\delta_1}}-\frac{C_{3}}{(1+R)^{\delta_2}},
\\
&K_{3}=\inf_{(r,y)\in\Gamma}
V(r,y)=C_{1}-\frac{C_{2}}{(1+2R)^{\delta_1}}-\frac{C_{3}}{(1+2R)^{\delta_2}},
\\
&C=2\beta+\frac12 \sigma^{2} \delta_2(\delta_2+1), \\
&C_1 = C_2 + C_3\,,
\end{align}
where $R\geq\varepsilon$, $\varepsilon>0$ is sufficiently small, 
$\delta_{1},\delta_{2} >0$ are positive constants
satisfying $(1+\delta_1)(1+\delta_2) = 2\gamma$,
and $C_{2},C_{3} >0$ are determined
separately for each case as follows 
(this is a restatement of the inequalities for $(a,b)$ following from
Lemma \ref{lemma:2}).

(i) For case (i) of Theorem~\ref{mainThm}, $C_{2},C_{3}>0$ satisfy the inequalities
\begin{equation}
R^{\delta_2(\delta_1+2)} \leq
\frac{\delta_2 C_3}{\delta_1 C_2 \sigma^2} \left( \frac{R}{1+R}\right)^{\delta_2-\delta_1}
\leq \frac{R^{2\gamma-\delta_1-1}-\kappa_1}{\kappa_2}\,,
\end{equation}
where $\kappa_1,\kappa_2$ are defined in (\ref{k1def}), (\ref{k2def}),
and $R$ is in the range allowed by condition (i) of Theorem~\ref{mainThm}.

(ii) For case (ii) of Theorem~\ref{mainThm}, $C_{2},C_{3}>0$ satisfy the inequality
\begin{eqnarray}
\frac{\delta_2 C_3}{\delta_1 C_2 \sigma^2} \left( \frac{R}{1+R}\right)^{\delta_2-\delta_1}
\leq \min\left\{ R^{\delta_2(\delta_1+2)}, 
\frac{\kappa_1}{R^{-\delta_2}-\kappa_2} \right\}\,,
\end{eqnarray}
where $\kappa_1,\kappa_2$ are defined in (\ref{k1def}), (\ref{k2def}),
and $R = \frac{1}{\delta_2}$.
\end{remark}

\begin{remark}\label{rem:C}
The constraints of Theorem~\ref{mainThm} can be made stronger by replacing
\begin{equation}
2\beta + \frac12\sigma^2 \delta_2(\delta_2+1) \to 
\max\{2\delta_1,\delta_2\}\beta + \frac12\sigma^2 \delta_2(\delta_2+1)\,.
\end{equation}
See the discussion around Eq.~(\ref{Csol}) about the choice of the constant $C$. 
This gives  a wider region for $\beta$. 
\end{remark}

\begin{figure}
\centering
  \includegraphics[width=4.0in]{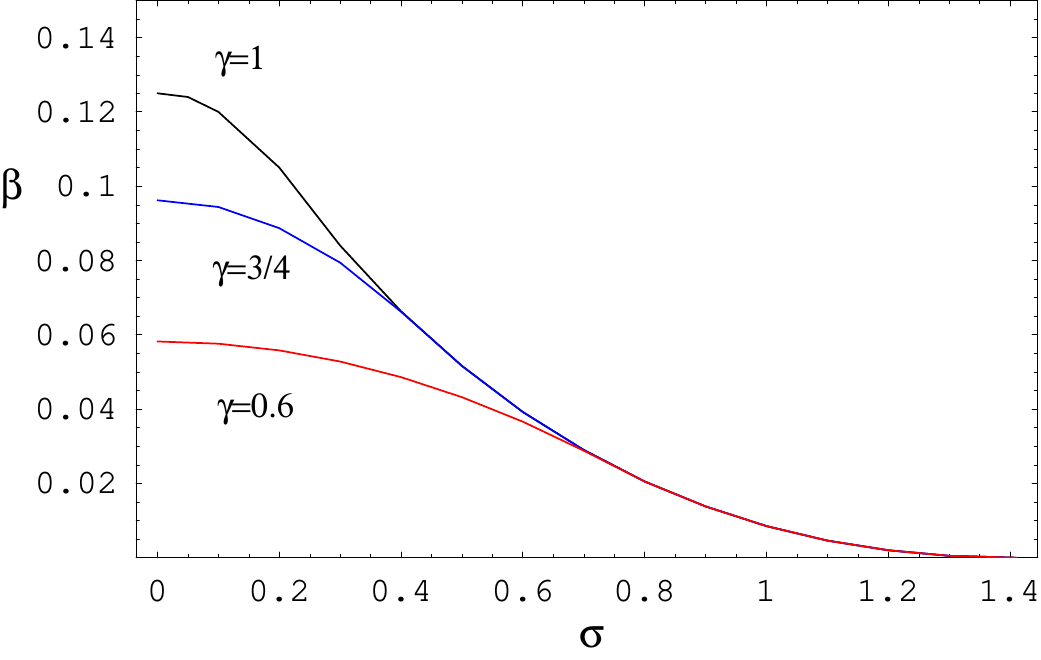}
\caption{Region in the $(\sigma,\beta)$ plane allowed by the condition (ii)
in Theorem~\ref{mainThm}.
For given $\gamma \in (\frac12,1]$, this condition is
satisfied for $\beta$ below the curves shown. }
\label{Fig:condG}
\end{figure}

\subsection{Numerical study}

We study here the regions for $(\sigma,\beta)$ allowed by Theorem~\ref{mainThm}.
We discuss only the condition (ii) which is more amenable to an analytical 
treatment. The resulting region for $(\beta,\sigma)$ includes all the typical 
values of these
parameters which are relevant for applications $0<\sigma < 1.0$ and $0<\beta< 0.1$,
see for example \cite{CaZhu}. The constraint on $\beta$ can be weakened further,
see Remark~\ref{rem:C}.

The condition (ii) of Theorem~\ref{mainThm} is satisfied in the region
below the curves shown in Figure \ref{Fig:condG}. 
For each $\frac12 < \gamma \leq 1$
there is one curve, corresponding to $\delta_2$ taking values in 
$0 < \delta_2 < 2\gamma-1$.

We outline the main steps in the derivation of these regions. The function
$G(R)= \frac{\delta_2 R}{(1+R)^{\delta_2+1}}$ with $\delta_2 >0$ has the 
following properties.

(i) $G(R)$ vanishes for $R\to 0$ and $R\to \infty$. The function $G(R)$
increases for $R< R_0(\delta_2)$ and decreases for
$R>R_0(\delta_2)$, with $R_0(\delta_2) =\frac{1}{\delta_2}$ .

(ii) $G(R)$ has a maximum at $R_0(\delta_2)$. At this point
the value of the function is
\begin{equation}
G(R_0) = \left( \frac{\delta_2}{1+\delta_2}\right)^{\delta_2+1}\,. 
\end{equation}

Fix the values of $\gamma$ and $\sigma$.
By scanning over $\delta_2 \in [0,2\gamma-1]$, find the maximum of
the expression 
\begin{equation}\label{delta2st}
\delta_{2*} := \arg\max_{\delta_2 \in [0,2\gamma-1]} 
\left\{G(R_0(\delta_2)) - \frac12\sigma^2 \delta_2(\delta_2+1)
\right\} \,.
\end{equation}
Then the values of $\beta$ allowed by the condition
(ii) of Theorem~\ref{mainThm}  (for given $\sigma,\gamma$) are
\begin{equation}\label{betarange}
0 \leq \beta \leq \frac12 G(R_0(\delta_{2*})) - \frac14 \sigma^2 \delta_{2*}
(\delta_{2*}+1)\,,
\end{equation}
where $\delta_{2*}$ is given by (\ref{delta2st}).
This region for $\beta$ is shown in Fig.~\ref{Fig:condG} for several
values of $\gamma$. The region becomes smaller as $\gamma$ approaches $\frac12$
and disappears at this point.

The value $\delta_{2*}$ decreases with $\sigma$, at fixed $\gamma$.
(Recall that this determines also the range of allowed values for $R$, 
which includes the point $R_0 = 1/\delta_{2*}$.)
In a range of sufficiently small $\sigma$, the maximum in (\ref{delta2st}) is
realized at the maximally allowed value $\delta_{2*}= 2\gamma-1$. 
In this region the curves for maximally allowed $\beta$ with different
values of $\gamma$ are distinct, as seen in Figure~\ref{Fig:condG}.
For $\sigma$ above a certain value, which depends on $\gamma$,
the value of $\delta_{2*}$ decreases from $2\gamma-1$ to zero.
In this region the maximal
$\beta$ curves are overlapping, since $\delta_{2*}$ is independent of $\gamma$.

There is a maximum value of $\sigma$ for which positive values of $\beta$ are
allowed. At this maximum value, which depends on $\beta$, 
$\delta_{2*}$ reaches zero. 
For $\beta=0$, this maximum value is $\sigma_{\rm max}=\sqrt2$. This follows
from the small-$\delta_2$ expansion 
\begin{equation}
G(R_0(\delta_2)) = \delta_2 + \delta_2^2(\log\delta_2+1) + O(\delta_2^3)\,.
\end{equation}
Substituting into (\ref{betarange}) gives
\begin{eqnarray}
\beta \leq \frac12 \delta_{2*}\left(1-\frac12\sigma^2\right) + O(\delta_{2*}^2)\,.
\end{eqnarray}
Requiring the cancellation of the $O(\delta_{2*})$ term gives the
maximal value $\sigma_{\rm max}=\sqrt2$.

\subsection{Almost sure explosion}

In Theorem \ref{mainThm}, we showed that
under certain conditions, the explosion occurs with positive
probability. Under some additional assumptions, 
one can further prove the almost sure explosion, that is, that explosion
occurs with probability one.

\begin{theorem}\label{asExplosionThm}
Suppose the assumptions of Theorem \ref{mainThm} are satisfied. 
Assume $\beta>0$.
For sufficiently large $r_{0}$ so that
\begin{equation}
r_{0}>\max\left\{\frac{e}{\beta}(4\beta R+\beta+\sigma^{2}),\frac{\sigma^{2}}{\beta}
e^{\frac{e^{2R}}{\sigma^{2}}(4\beta R+\beta+\sigma^{2})-2R-1}\right\},
\end{equation}
where $R$ is determined such that either of the conditions (i) or (ii) of
Theorem \ref{mainThm} holds,
we have the almost sure explosion, that is, $\mathbb{P}(\tau<\infty)=1$.
\end{theorem}

\begin{remark}
Under the assumptions in Theorem \ref{asExplosionThm}, 
$K_{0}$ from the conditions in Proposition \ref{Thm2} is given by
\begin{equation}
K_{0}=\min\left\{C_{1}-\frac{C_{2}}{(1+R)^{\delta_{1}}}-C_{3},
C_{1}-C_{2}-\frac{C_{3}}{(1+R)^{\delta_{2}}}\right\},
\end{equation}
where $C_{1},C_{2},C_{3},R,\delta_{1},\delta_{2}$ are defined in Remark \ref{remarkCondition}.
\end{remark}

\section{Implications for zero coupon bond prices and Eurodollar futures}\label{sec:imply}

The explosion of $(r_{t},y_{t})$
is equivalent to the explosion of $r_{t}$ due to the explicit form of $y_{t}$ in terms of $(r_{s})_{0\leq s\leq t}$.
The explosion of $r_{t}$ thus implies that the prices of zero
coupon bonds $P(t,T)$ become zero almost surely for all $t>\tau$, with $\tau$ 
the explosion time of $r_t$.

This  follows from Eq.~(\ref{ZCB}) for the zero coupon bond price, which gives
\begin{equation}\label{ZCB2}
P(T,T+\delta) = \frac{P(0,T+\delta)}{P(0,T)}
\exp\left(- G(T,T+\delta) x_T - \frac12 G^2(T,T+\delta) y_T \right)\,,
\end{equation}
where we recall that $x_{T}=r_{T}-\lambda(T)$.
Suppose the assumptions in Theorem~\ref{mainThm}
are satisfied, then $\mathbb{P}(\tau<\infty)>0$,
which implies that for sufficiently large $T$, $\mathbb{P}(\tau<T)>0$.
As a result, with positive probability, and sufficiently large $T$,
the zero coupon bond price $P(T,T+\delta)$ collapses to zero.

This implies that interest rates $L(T_1,T_2)$ explode for all $T_1>\tau$.
Recall that the rate $L(T_1,T_2)$ is related to $P(T_1,T_2)$ as
$L(T_1,T_2) = \frac{1}{T_2-T_1}(P^{-1}(T_1,T_2) - 1)$, see e.g. \cite{APbook}.

The prices of any derivatives depending on $L(T_1,T_2)$ such as interest rate
caps, swaptions, CMS swaps, and Eurodollar futures also become infinite. 
We will show this explicitly for the prices of Eurodollar futures contracts.
Using Eq.~(\ref{ZCB2}) for the zero coupon bond price $P(T,T+\delta)$ we get
\footnote{Note that in Eq.~(43) in \cite{ORL}, there is
a typo: the factor on the
right-hand side of this equation
 should be $\frac{P(0,T)}{P(0,T+\delta)}$.}
\begin{equation}
\mathbb{E}^{\mathbb{Q}}[P^{-1}(T,T+\delta)] = \frac{P(0,T)}{P(0,T+\delta)}
\mathbb{E}^{\mathbb{Q}}\left[\exp\left(G(T,T+\delta) x_T + \frac12 G^2(T,T+\delta) y_T \right)\right]\,.
\end{equation} 
Suppose the assumptions in Theorem~\ref{mainThm}
are satisfied, then $x_T=y_T=\infty$ with positive probability
for sufficiently large $T$.
It follows that the Eurodollar futures price explodes to infinity,
that is, $\mathbb{E}^{\mathbb{Q}}[P^{-1}(T,T+\delta)]=\infty$,
for sufficiently large $T$.

In practical applications the explosions of the short rate
$r_t$ could be avoided by capping the short rate volatility to a finite
value $c$, possibly
using a prescription of the same type as that proposed in \cite{HJM},
$\sigma_r(r_t) \to \min\{\max\{0,\sigma_r(r_t)\},c\}$. 
With this change, the diffusion coefficients
satisfy the sub-linear growth condition of Theorem 5.2.9 in 
\cite{KaratzasShreve}, which ensures that the solution $r_t$ exists and is
non-explosive.


\section{Appendix: Proofs}\label{sec:proofs}

\begin{proof}[Proof of Theorem~\ref{mainThm}]

(a) 
For $0<\gamma \leq \frac12$ the coefficients of the
2-d diffusion \eqref{rSDE2}, \eqref{ySDE2} satisfy the conditions of Theorem
5.2.9 in \cite{KaratzasShreve}, which we recall here briefly for convenience. 

Consider the SDE for the $d-$dimensional vector $X_t \in\mathbb{R}^d$
\begin{equation}\label{SDE1}
X_t = \sigma(t,x) dW_t + b(t,x) dt 
\end{equation}
where $x\in \mathbb{R}^d$ and 
$W_t$ is a $d$-dimensional Brownian motion. Assume that the coefficients
satisfy the global Lipschitz and linear growth conditions 
\begin{eqnarray}\label{Thm529}
&& |\!| b(t,x) - b(t,y) |\!| + 
|\!| \sigma(t,x) - \sigma(t,y) |\!|  \leq K |\!| x-y |\!|\,, \\
&& |\!| b(t,x) |\!|^2 + 
|\!| \sigma(t,x) |\!|^2 \leq K^2 (1 + |\!|x |\!|^2)\,,\nonumber
\end{eqnarray}
for every $0 \leq t< \infty, x\in \mathbb{R}^d, y\in \mathbb{R}^d$ and
$K$ is a positive constant. 
Under these conditions, there exists a continuous, adapted process
$X = \{X_t; 0 \leq t < \infty\}$ which is a strong solution of the SDE
(\ref{SDE1}) with initial condition $X_0$, and is furthermore square-integrable.

The SDE \eqref{rSDE2}, \eqref{ySDE2} with $0<\gamma \leq \frac12$ satisfies the 
conditions (\ref{Thm529}), and thus $\{r_t, y_t\}_{t\geq 0}$ does not explode.
We study next the case $\frac{1}{2}<\gamma\leq 1$, where the linear growth
condition does not hold.

(b) 
The boundary $r_{t}=0$, $y_{t}=0$ is unattainable. Indeed, for $r_t<\varepsilon$,
we have $dr_t = \sigma r_t dW_t + (y_t - \beta r_t + \beta r_0)dt $ with
\begin{equation}\label{43}
y_{t}=\sigma^{2}\int_{0}^{t}r_{s}^{2} \min(r_s^{2(\gamma-1)},
\varepsilon^{2(\gamma-1)})
e^{2\beta(s-t)}ds>0,
\end{equation}
and since the term $\beta r_{0}>0$ and $y_{t}>0$ in the drift term of $r_{t}$, 
by comparing $r_{t}$ with a geometric Brownian motion, we have $r_{t}>0$. 
Therefore, $(y_{t},r_{t})\in\mathbb{R}^{+}\times\mathbb{R}^{+}$. 
This generalizes to $\gamma \in (\frac12,1]$ the result of 
\eqref{rpositive} and \eqref{ypositive}. Although this is proved here for the time-homogeneous
case $\lambda(t)=\lambda_0$, the result is easily seen to hold also under
the weaker assumption $\lambda'(t) + \beta \lambda(t)\geq 0$
by a comparison argument.

Let us take $\mathcal{D}:=(0,R)\times(0,R)$, with $R\geq\varepsilon$,
where we define $\mathcal{D}^{c}=\mathbb{R}^{+}\times\mathbb{R}^{+}\backslash\mathcal{D}$.
It is clear that $\mathcal{D}$ is a bounded open set.
The boundary $\partial\mathcal{D}$
is regular since $\beta r_{0}>0$ in the drift term of $r_{t}$. 
It is also easy to see that, 
for $\frac{1}{2}<\gamma\leq 1$, on any compact subset of 
$\mathbb{R}^{+}\times\mathbb{R}^{+}$, 
the coefficients of the SDE \eqref{rSDE2},\eqref{ySDE2}, 
are continuous and Lipschitz.

Assume the following form for the Lyapunov function 
\begin{equation}\label{Vdef}
V(r,y)= C_1 - \frac{C_2}{(1+y)^{\delta_1}} - \frac{C_3}{(1+r)^{\delta_2}}\,,
\end{equation}
with $C_1,C_2,C_3>0$ and $\delta_1,\delta_2 >0$ are positive constants 
satisfying the condition
\begin{equation}\label{d12cond}
(1+\delta_1)(1+\delta_2)=2\gamma \,.
\end{equation}

We would like to test the conditions (A.1),(A.2) and (A.3) of 
Proposition~\ref{Thm1}.

(1) \textit{Condition (A.1). } 
For any $(r,y) \in \mathcal{D}^c$, we have
\begin{equation}
V(r,y) 
\geq\min\left\{C_{1}-\frac{C_{2}}{(1+R)^{\delta_{1}}}-C_{3},
C_{1}-C_{2}-\frac{C_{3}}{(1+R)^{\delta_{2}}}\right\}>0\,,
\end{equation}
provided that we take
\begin{equation}\label{HoldI}
C_1 \geq C_2+C_3\,.
\end{equation}
Thus $V(r,y)$ defined on $\mathcal{D}^c$ is a positive function.
Since $r,y>0$, it is clear that $V(r,y)\leq C_1$. Thus the condition
(A.1) is satisfied.

(2) \textit{Condition (A.2). } 
Note that $\partial\mathcal{D}=
\{(R,y):0\leq y\leq R\}\cup\{(r,R):0\leq r\leq R\}$.
Thus, we have
\begin{equation}
K_{2}=\sup_{(r,y)\in\partial\mathcal{D}}V(r,y)=
C_{1}-\frac{C_{2}}{(1+R)^{\delta_1}}-\frac{C_{3}}{(1+R)^{\delta_2}}.
\end{equation}

On the other hand, taking 
$\Gamma=[2R,\infty)\times[2R,\infty)\subset\mathcal{D}^{c}$, 
we obtain
\begin{equation}
K_{3}=\inf_{(r,y)\in\Gamma}V(r,y)=
C_{1}-\frac{C_{2}}{(1+2R)^{\delta_1}}-\frac{C_{3}}{(1+2R)^{\delta_2}}
>K_{2}.
\end{equation}
Hence, the condition (A.2) holds.

(3) \textit{Condition (A.3).} 
Finally, let us check the condition (A.3). Note that
\begin{eqnarray}\label{Leps2}
&& \mathcal{L}_\varepsilon V(r,y) =
\delta_1 C_{2}\sigma^{2}
\frac{\min(r^{2\gamma}, r^2 \varepsilon^{2\gamma-2})}{(1+y)^{\delta_1+1}}
-2\delta_1 C_{2}\beta \frac{y}{(1+y)^{\delta_1+1}}
+\delta_2 C_{3}\frac{y}{(1+r)^{\delta_2+1}} \\
&& \qquad -\delta_2 C_{3}\frac{\beta r}{(1+r)^{\delta_2+1}}
+\delta_2 C_{3}\beta r_{0}\frac{1}{(1+r)^{\delta_2+1}} 
-\frac12 \delta_2(\delta_2+1) C_{3}\sigma^{2}
\frac{\min(r^{2\gamma}, r^2 \varepsilon^{2\gamma-2})}{(1+r)^{\delta_2+2}}.
\nonumber
\end{eqnarray}
Therefore, 
\begin{align*}
&\mathcal{L}_\varepsilon V-CV
\\
&=\delta_1 C_{2}\sigma^{2}
\frac{\min(r^{2\gamma}, r^2 \varepsilon^{2\gamma-2})}{(1+y)^{\delta_1+1}}
+\delta_2 C_{3}\frac{y}{(1+r)^{\delta_2+1}}
\\
&\qquad
+\frac{C_{2}}{(1+y)^{\delta_1}}\left[C-\frac{2\beta\delta_1 y}{1+y}\right]
+\frac{C_{3}}{(1+r)^{\delta_2}}\left[C-\frac{\delta_2\beta r}{1+r}
-\frac12 \delta_2(\delta_2+1) \sigma^{2}
\frac{\min(r^{2\gamma}, r^2 \varepsilon^{2\gamma-2})}{(1+r)^{2}}\right]
\\
&\qquad\qquad\qquad
+\delta_2 C_{3}\frac{\beta r_0}{(1+r)^{\delta_2+1}}-CC_{1}
\\
&\geq
\delta_1 C_{2}\sigma^{2}
 \frac{\min(r^{2\gamma}, r^2 \varepsilon^{2\gamma-2})}{(1+y)^{\delta_1+1}}
+\delta_2 C_{3}\frac{y}{(1+r)^{\delta_2+1}}
\\
&\qquad
+\frac{C_{2}}{(1+y)^{\delta_1}}\left[C-2\beta \delta_1 \right]
+\frac{C_{3}}{(1+r)^{\delta_2}}
\left[C-\delta_2 \beta-\frac12 \sigma^{2} \delta_2(\delta_2+1) 
 \frac{\min(r^{2\gamma}, r^2 \varepsilon^{2\gamma-2})}{(1+r)^2}\right]
\\
&\qquad\qquad\qquad
+\delta_2 C_{3}\frac{\beta r_0}{(1+r)^{\delta_2+1}}-CC_{1}.
\end{align*}

Furthermore, for $1<2\gamma \leq 2$, we have
\begin{eqnarray}
&& 0 < \frac{r^{2\gamma}}{(1+r)^2} \leq 1\,, \qquad r \geq 0\,, \\
&& 0 < \frac{r^{2} \varepsilon^{2\gamma-2} }{(1+r)^2} \leq 1\,, \qquad
 0 \leq r \leq \varepsilon\,,
\end{eqnarray}
($\varepsilon$ is small, say less than $1$) such that we have
\begin{align*}
\mathcal{L}_\varepsilon V-CV 
&\geq
\delta_1 C_{2}\sigma^{2}
 \frac{\min(r^{2\gamma}, r^2 \varepsilon^{2\gamma-2})}{(1+y)^{\delta_1+1}}
+\delta_2 C_{3}\frac{y}{(1+r)^{\delta_2+1}}
\\
&\qquad
+\frac{C_{2}}{(1+y)^{\delta_1}}\left[C-2\beta \delta_1 \right]
+\frac{C_{3}}{(1+r)^{\delta_2}}
\left[C-\delta_2 \beta-\frac12 \sigma^{2} \delta_2(\delta_2+1) \right]
\\
&\qquad\qquad\qquad
+\delta_2 C_{3}\frac{\beta r_0}{(1+r)^{\delta_2+1}}-CC_{1}.
\end{align*}

Let us choose $C$ to be a fixed constant so that
\begin{equation}\label{HoldII}
C\geq\max\left\{2\delta_1 \beta,\delta_2\beta+
\frac12 \sigma^{2}\delta_2(\delta_2+1)\right\}.
\end{equation}
Then we have
\begin{equation*}
\mathcal{L}_\varepsilon V-CV
\geq
\delta_1 C_{2}\sigma^{2}
 \frac{\min(r^{2\gamma}, r^2 \varepsilon^{2\gamma-2})}{(1+y)^{\delta_1+1}}
+\delta_2 C_{3}\frac{y}{(1+r)^{\delta_2+1}}
-CC_{1}.
\end{equation*}
Recall that for $(r,y)\in\mathcal{D}^{c}$, we have either $y\geq R$ or $r\geq R$,
and we chose $\varepsilon < R$.

(I) 
If $y\geq R$ and $r<R$, then we have for both $0\leq r \leq\varepsilon$
and $\varepsilon < r < R$, by positivity of the first term,
\begin{align*}
\mathcal{L}_\varepsilon V-CV
&\geq
\delta_1 C_{2}\sigma^{2}
\frac{\min(r^{2\gamma}, r^2 \varepsilon^{2\gamma-2})}{(1+y)^{\delta_1+1}}
+\delta_2 C_{3}\frac{y}{(1+r)^{\delta_2+1}}
-CC_{1}
\\
&\geq
\delta_2 C_{3}\frac{R}{(1+R)^{\delta_2+1}}-CC_{1}.
\end{align*}

(II) If $r\geq R$ and $y<R$, then we have 
(since $r\geq  R> \varepsilon$)
\begin{align*}
\mathcal{L}_\varepsilon V-CV
&\geq
\delta_1 C_{2}\sigma^{2}\frac{r^{2\gamma}}{(1+y)^{\delta_1+1}}
+\delta_2 C_{3}\frac{y}{(1+r)^{\delta_2+1}}
-CC_{1}
\\
&\geq
\delta_1 C_{2}\sigma^{2}\frac{ R^{2\gamma}}{(1+R)^{\delta_1+1}}-CC_{1}.
\end{align*}

(III) If $y\geq R$ and $r\geq R$, then we have
(again by $r\geq R> \varepsilon$)
\begin{align}
\mathcal{L}_\varepsilon V-CV
&\geq
\delta_1 C_{2}\sigma^{2}\frac{ r^{2\gamma}}{(1+y)^{\delta_1+1}}
+\delta_2 C_{3}\frac{y}{(1+r)^{\delta_2+1}}
-CC_{1}
\nonumber
\\
&\geq\delta_1 C_{2}\sigma^{2}\left( \frac{R}{1+R} \right)^{\delta_1+1}
\frac{ r^{2\gamma}}{y^{\delta_1+1}}
+\delta_2 C_{3} \left(\frac{R}{1+R}\right)^{\delta_2+1}
\frac{y}{r^{\delta_2+1}}
-CC_{1} 
\nonumber
\\
& = \delta_1 C_{2}\sigma^{2}\left( \frac{R}{1+R} \right)^{\delta_1+1}
x^{\delta_1+1}
+\delta_2 C_{3} \left(\frac{R}{1+R}\right)^{\delta_2+1}
\frac{1}{x}
-CC_{1}\,,
\label{LVCV}
\end{align}
where we denoted
\begin{equation}
x = \frac{r^{\delta_1+1}}{y}\,.
\end{equation}
The condition $2\gamma = (1+\delta_1)(1+\delta_2)$ was used to reduce the
dependence  on $(r,y)$ to a function of $x$ in the last step.

The sum of the first two terms is bounded from below by the following Lemma.

\begin{lemma}\label{minF}
The infimum of the function $\hat{F}(x): (0,\infty) \to (0,\infty)$ defined as
\begin{equation}\label{fdef}
\hat{F}(x) := a x^{\delta_1+1} + \frac{b}{x}\,,\quad
a,b > 0
\end{equation}
is given by
\begin{equation}
\inf_{x>0}\hat{F}(x)=\kappa_\delta a^{\frac{1}{\delta_1+2}} 
b^{\frac{\delta_1+1}{\delta_1+2}}\,,
\end{equation}
where
\begin{equation}
\kappa_\delta = (\delta_1+2)(\delta_1+1)^{-\frac{\delta_1+1}{\delta_1+2}}>1\,.
\end{equation}
\end{lemma} 

\begin{proof}[Proof of Lemma \ref{minF}]
The minimum of $\hat{F}(x)$ is achieved at $\hat{F}'(x)=0$, which gives
\begin{equation}
\hat{F}'(x)=a(\delta_{1}+1)x^{\delta_{1}}-\frac{b}{x^{2}}=0\,,
\end{equation}
which implies that
\begin{equation}
\min_{x>0}\hat{F}(x)=(\delta_1+2)(\delta_1+1)^{-\frac{\delta_1+1}{\delta_1+2}}a^{\frac{1}{\delta_1+2}} 
b^{\frac{\delta_1+1}{\delta_1+2}}\,.
\end{equation}
To see that $\kappa_\delta$ is larger than $1$
for $\delta_1 \in [0,1]$, 
let us define $\hat{G}(\delta_{1}):=\log(\delta_{1}+2)-\frac{\delta_{1}+1}{\delta_{1}+2}\log(\delta_{1}+1)$.
We will show that $\hat{G}$ is decreasing in $\delta_{1}\in[0,1]$, that 
is due to
\begin{equation}
\hat{G}'(\delta_{1})= - \frac{\log(\delta_1+1)}{(\delta_1+2)^2}  <0 
\end{equation}
Hence, $\kappa_{\delta}=e^{\hat{G}(\delta_{1})}$ is decreasing in $\delta_{1}\in[0,1]$.
We can compute that at $\delta_{1}=1$, 
we have $\kappa_\delta = (1+2)(1+1)^{-\frac{1+1}{1+2}}=\frac{3}{2^{\frac{2}{3}}}>1$.
Hence, $\kappa_{\delta}$ is larger than $1$.
\end{proof}

Therefore, following \eqref{LVCV}, we have 
\begin{equation*}
\mathcal{L}_\varepsilon V-CV
\geq\kappa_{\delta}
\left(\delta_1 C_{2}\sigma^{2}\left(\frac{R}{1+R} \right)^{\delta_1+1}\right)^{\frac{1}{\delta_1+2}}
\left(\delta_2 C_{3} \left(\frac{R}{1+R} \right)^{\delta_2+1}\right)^{\frac{\delta_1+1}{\delta_1+2}}
-CC_{1}.
\end{equation*}

Hence, from (I), (II) and (III), we conclude that
$\mathcal{L}_\varepsilon V\geq CV$ for any $(r,y)\in\mathcal{D}^{c}$ if we have
\begin{align}\label{HoldIII}
CC_{1}&\leq
\min\Bigg\{\delta_2 C_{3}\frac{R}{(1+R)^{\delta_2+1}},
\delta_1 C_{2}\sigma^{2}\frac{ R^{2\gamma}}{(1+R)^{\delta_1+1}},
\\
&\qquad\qquad
\kappa_{\delta}
\left(\delta_1 C_{2}\sigma^{2}\left(\frac{R}{1+R}\right)^{\delta_1+1} 
\right)^{\frac{1}{\delta_1+2}}
\left(\delta_2 C_{3}\left(\frac{R}{1+R} \right)^{\delta_2+1}\right)^{\frac{\delta_1+1}{\delta_1+2}}
\Bigg\}.
\nonumber
\end{align}

To summarize, in order to have the Lyapunov function $V(r,y)$ to be bounded, 
positive
and satisfy (A.1), (A.2), (A.3), we need the conditions \eqref{HoldI}, 
\eqref{HoldII} and \eqref{HoldIII} to hold simultaneously. 
Taking
\begin{eqnarray}
&& C_{1}=C_{2}+C_{3}\,,  \\
\label{Csol}
&& C=\max\left\{ 2\delta_1, \delta_2 \right\} \cdot\beta+\frac12 \sigma^{2}\delta_2(\delta_2+1)\,,
\end{eqnarray}
then \eqref{HoldI} and \eqref{HoldII} are satisfied. 
This can be simplified by replacing $\max\{2\delta_1, \delta_2 \} \to 2$ since
for all $0<\gamma \leq 1$ we have $\delta_{1}, \delta_{2} \leq 1$.

The condition \eqref{HoldIII} is satisfied as well if we have
\begin{align}\label{19}
&\left(2\beta+\frac12 \sigma^{2} \delta_2 (\delta_2+1)\right)(C_{2}+C_{3})
\\
&\leq
\min\Bigg\{\delta_2 C_{3}\frac{R}{(1+R)^{\delta_2+1}},
\delta_1 C_{2}\sigma^{2}\frac{ R^{2\gamma}}{(1+R)^{\delta_1+1}},
\nonumber \\
&\qquad\qquad
\kappa_{\delta}
\left(\delta_1 C_{2}\sigma^{2}\left(\frac{R}{1+R}\right)^{\delta_1+1} 
\right)^{\frac{1}{\delta_1+2}}
\left(\delta_2 C_{3}\left(\frac{R}{1+R} \right)^{\delta_2+1}\right)^{\frac{\delta_1+1}{\delta_1+2}}
\Bigg\}.
\nonumber
\end{align}

\textit{The study of the inequality (\ref{19}).}
We study next the conditions for $(\beta,\sigma)$ for which 
the inequality (\ref{19}) is satisfied in a region of $(C_2,C_3)$,
at least for one value of $R$.
We start by writing it in an equivalent way as
\begin{equation}\label{19p}
\kappa_1 a + \kappa_2 b \leq \min 
\left\{\kappa_\delta a^{\varepsilon_1} b^{\varepsilon_2}, 
a R^{2\gamma-\delta_1-1}, b R^{-\delta_2} \right\}\,,
\end{equation}
where 
\begin{equation}
\varepsilon_1 = \frac{1}{\delta_1+2} \,,\quad
\varepsilon_2 = \frac{\delta_1+1}{\delta_1+2} 
\end{equation}
satisfying $\varepsilon_1+\varepsilon_2=1$, 
and we defined the new variables
\begin{equation}
a = \delta_1 C_2 \sigma^2 \left( \frac{R}{1+R}\right)^{\delta_1+1} \,,\quad
b = \delta_2 C_3 \left( \frac{R}{1+R}\right)^{\delta_2+1}\,,
\end{equation}
and denoted the constants
\begin{eqnarray}
&& \kappa_\delta := (\delta_1+2) (\delta_1+1)^{-\frac{\delta_1+1}{\delta_1+2}}\,, \\
\label{k1def}
&& \kappa_1 := \frac{1}{\delta_1 \sigma^2} 
\left(2\beta+\frac12\sigma^2 \delta_2(\delta_2+1) \right) 
   \left( \frac{1+R}{R} \right)^{\delta_1+1}\,, \\
\label{k2def}
&& \kappa_2 := \frac{1}{\delta_2} 
\left(2\beta+\frac12\sigma^2\delta_2(\delta_2+1)\right) 
   \left( \frac{1+R}{R} \right)^{\delta_2+1}\,.
\end{eqnarray}

We would like to obtain the region in the $(a,b) \in \mathbb{R}_+^2$ plane
where the inequality (\ref{19p}) holds, and find conditions on $\kappa_1,\kappa_2$
(or equivalently $\beta,\sigma)$) for which this region is non-empty, 
at least for one value of $R$. These regions are given by the following Lemma.

\begin{lemma}\label{lemma:2}
The inequality 
\begin{equation}
\kappa_1 a + \kappa_2 b \leq \min\left\{\kappa_\delta a^{\varepsilon_1} b^{\varepsilon_2},
a R^{2\gamma-\delta_1-1}, b R^{-\delta_2}\right\}\,,
\end{equation} 
with $\kappa_{1}, \kappa_{2}, \varepsilon_1, \varepsilon_2 > 0, 
\kappa_\delta > 1, \varepsilon_1 (\delta_1+2)=1$ and 
$\varepsilon_1+\varepsilon_2=1$,
holds in two regions of the $(a,b)$ plane:

(i) A wedge-like
region of the positive quadrant of the $(a,b)$ plane, contained between 
the two straight lines passing through origin
\begin{equation}\label{2lines}
a R^{\delta_2(\delta_1+2)} \leq b \leq a 
\frac{R^{2\gamma-\delta_1-1}-\kappa_1}{\kappa_2}\,.
\end{equation}

(ii) A wedge-like region of the positive quadrant of the $(a,b)$ plane, 
below a straight line passing through origin given by
\begin{equation}\label{1line}
b \leq a \min \left\{ R^{\delta_2(\delta_1+2)}, 
\frac{\kappa_1}{R^{-\delta_2}-\kappa_2} \right\} \,.
\end{equation}
\end{lemma}

\begin{proof}[Proof of Lemma~\ref{lemma:2}]
We prove that the inequality (\ref{19p}) is satisfied in the regions
(\ref{2lines}) and (\ref{1line}).

The line 
\begin{equation}
b = a R^{2\gamma-\delta_1+\delta_2-1} = a R^{\delta_2(\delta_1+2)} 
\end{equation}
divides the first quadrant of the $(a,b)\in \mathbb{R}^2_+$ plane
into two regions: 
\begin{itemize}
\item[(i)] 
\textit{Region 1} with $b >  a R^{\delta_2(\delta_1+2)}$; 
\item[(ii)] 
\textit{Region 2} with $b <  a R^{\delta_2(\delta_1+2)}$. 
\end{itemize}

We show that the inequality (\ref{19p}) simplifies in each
of these regions as follows.

(i) \textit{Region 1} with $b > a R^{\delta_2(\delta_1+2)}$.

In this region we have clearly 
\begin{equation}
\min\left\{b R^{-\delta_2}, a R^{2\gamma-\delta_1-1}\right\}= a R^{2\gamma-\delta_1-1}\,.
\end{equation} 
Furthermore, we have
\begin{align}
\kappa_\delta a^{\varepsilon_1} b^{\varepsilon_2} 
&>\kappa_\delta a R^{\varepsilon_2 \delta_2 (\delta_1+2)} =
\kappa_\delta a R^{\delta_2(\delta_1+1)} 
\\
&= \kappa_\delta a R^{\delta_2(2\gamma-\delta_1-1)} \geq  
a R^{\delta_2(2\gamma-\delta_1-1)}\,, \nonumber
\end{align}
since $\kappa_\delta>1$ as noted above. 
Thus the inequality (\ref{19p}) reduces in this region
to a linear inequality
\begin{equation}
\kappa_1 a + \kappa_2 b \leq a R^{2\gamma-\delta_1-1}\,.
\end{equation}
This gives the upper bound on $b$ in (\ref{2lines}).

(ii) \textit{Region 2} with $b < a R^{\delta_2(\delta_1+2)}$. 

In this region we have clearly 
\begin{equation}
\min\left\{b R^{-\delta_2}, a R^{2\gamma-\delta_1-1}\right\}= b R^{-\delta_2}\,.
\end{equation}
Furthermore, we have the lower bound
\begin{equation}
\kappa_\delta a^{\varepsilon_1} b^{\varepsilon_2} > 
\kappa_\delta (b R^{-\delta_2(\delta_1+2)})^{\varepsilon_1} b^{\varepsilon_2} =
\kappa_\delta b R^{-\varepsilon_1\delta_2(\delta_1+2)} > b R^{-\delta_2}
\end{equation}
since $\kappa_\delta>1$. 
Thus the inequality (\ref{19p}) reduces in this region
to the linear inequality
\begin{equation}
\kappa_1 a + \kappa_2 b \leq b R^{-\delta_2}\,.
\end{equation}
This gives an upper bound on $b$
\begin{equation}
b \leq \frac{\kappa_1}{R^{-\delta_2}-\kappa_2} a
\end{equation}
which is useful only if $R^{-\delta_2} > \kappa_2$, or equivalently if
\begin{equation}
\left(2\beta+\frac12 \sigma^2 \delta_2 (\delta_2+1) \right) (1+R)^{\delta_2+1} 
< R \delta_2\,.
\end{equation}
This is obtained using the expression (\ref{k2def}) for $\kappa_{2}$.

If this condition is satisfied, then we get that (\ref{19p}) is satisfied
in the subset of region 2
\begin{equation}
b \leq \min \left\{R^{\delta_2(\delta_1+2)}, 
\frac{\kappa_1}{R^{-\delta_2}-\kappa_2} \right\}a\,.
\end{equation}

This is either the entire region 2, or a subset, bounded by the real axis
and the line $b=\frac{\kappa_1}{R^{-\delta_2}-\kappa_2}a$.
\end{proof}

Finally, let us get back to the proof of Theorem~\ref{mainThm}.
In order for the region (\ref{2lines}) to be non-empty, 
the following inequality must hold
\begin{equation}\label{ineq1}
\kappa_2 R^{\delta_2(\delta_1+2)} \leq R^{2\gamma-\delta_1-1}- \kappa_1\,.
\end{equation}
Substituting here the expressions (\ref{k1def}), (\ref{k2def}) 
for $\kappa_{1},\kappa_{2}$, this becomes
\begin{equation}
R^{2\gamma} \geq 
\left(2\beta+\frac12\sigma^2 \delta_2(\delta_2+1)\right)
\left(
\frac{1}{\delta_1 \sigma^2}
(1+R)^{\delta_1+1} +
\frac{1}{\delta_2}  
R^{\delta_1\delta_2+\delta_1+\delta_2} (1+R)^{\delta_2+1} \right)\,.
\end{equation}

In order for the region (\ref{1line}) to be non-empty one requires
$R^{-\delta_2} > \kappa_2$ which gives the inequality
\begin{equation}\label{ineq2}
R \geq \frac{1}{\delta_2} \left(2\beta+\frac12 \sigma^2\delta_2(\delta_2+1)\right) 
(1+R)^{\delta_2+1} \,.
\end{equation}

The inequality (\ref{ineq1}) yields the statement (i) of Theorem~\ref{mainThm},
and the inequality (\ref{ineq2}) the statement (ii). This completes the proof
of Theorem~\ref{mainThm}.
\end{proof}


\begin{proof}[Proof of Theorem~\ref{asExplosionThm}]

We would like to test the conditions (A.4) and (A.5) of 
Proposition~\ref{Thm2}.

(1) \textit{Condition (A.4).} 
Let $V(r,y)$ be the Lyapunov function (\ref{Vdef})
defined in our Theorem \ref{mainThm}. We have to check that its infimum on
$\mathcal{D}^{c}$ is positive. 
We can compute that
\begin{align}
K_{0}:=\inf_{(r,y)\in\mathcal{D}^{c}}V(r,y)
&=
C_1 - \sup_{(r,y)\in\mathcal{D}^{c}} \left(
\frac{C_2}{(1+y)^{\delta_1}} + \frac{C_3}{(1+r)^{\delta_2}} \right) 
\\
&=\min\left\{C_{1}-\frac{C_{2}}{(1+R)^{\delta_{1}}}-C_{3},
C_{1}-C_{2}-\frac{C_{3}}{(1+R)^{\delta_{2}}}\right\}\,.
\nonumber
\end{align}
By (\ref{HoldI}) we have $C_1 \geq C_2 +C_3$ which gives $K_0>0$.
Thus, (A.4) holds.

(2) \textit{Condition (A.5)}. 
According to Theorem 3.9. and the discussion at the beginning of Chapter 3.7. 
in \cite{Khasminskii}, 
it suffices to show that there exists a non-negative function $V_{0}(r,y)$ 
for $(r,y)\in\Gamma^{c}$ that is twice differentiable
in $(r,y)$ such that
\begin{equation}
\mathcal{L}_\varepsilon V_{0}(r,y)\leq-\alpha\,, \qquad\text{for any $(r,y) \in \Gamma^{c}$},
\end{equation}
where $\alpha>0$ is some constant. We use the notation $V_{0}$ to distinguish it
from the Lyapunov function $V$ defined in Theorem \ref{mainThm}.

Let us recall from the proof of Theorem \ref{mainThm} that 
$\Gamma=[2R,\infty)\times[2R,\infty)$.
Therefore, $\Gamma^{c}=\{(r,y): 0<y<2R\text{ or }0<r<2R\}$.
Let us define
\begin{equation}
V_{0}(r,y)=e^{-r}+e^{-y}.
\end{equation}
Then $V_{0}$ is non-negative and twice differentiable. 
We can compute that
\begin{equation}
\mathcal{L}_\varepsilon V_{0}(r,y)=
(-\sigma^{2}\min(r^{2\gamma}, r^2 \varepsilon^{2\gamma-2}) + 2\beta y) e^{-y}
+\left(\beta r-y-\beta r_{0}+\frac{1}{2}\sigma^{2}
\min(r^{2\gamma}, r^2 \varepsilon^{2\gamma-2})\right)e^{-r}.
\end{equation}

For $(r,y)\in\Gamma^{c}$, either one of the inequalities 
$0<y<2R$ or $0<r<2R$ holds. 

(a) If $0<r<2R$, 
we distinguish between $0< r \leq \varepsilon$ and 
$\varepsilon < r < 2R$. In the latter case we have
\begin{align*}
&\mathcal{L}_\varepsilon V_{0}(r,y)
\\
&\leq-\beta r_{0}e^{-2R}
+\sup_{0<y<\infty,\varepsilon <r<2R}
\left\{(-\sigma^{2}r^{2\gamma}+2\beta y)e^{-y}+
 \left(\beta r-y+\frac{1}{2}\sigma^{2}r^{2\gamma}\right)e^{-r}\right\}
\\
&\leq-\beta r_{0}e^{-2R}
+\sup_{0<y<\infty}2\beta ye^{-y}+
 \sup_{\varepsilon <r<2R}
 \left(\beta r+\frac{1}{2}\sigma^{2}r^{2\gamma}\right)e^{-r}
\\
&\leq-\beta r_{0}e^{-2R}+\frac{2\beta}{e}+(2\beta R+2\sigma^{2}R^{2\gamma})<0,
\end{align*}
for any sufficiently large $r_{0}$ such that
\begin{equation}\label{r0R}
r_{0}>\frac{e^{2R}}{\beta}\left[\frac{2\beta}{e}+(2\beta R+2\sigma^{2}R^{2\gamma})\right],
\end{equation}
where we assumed that $\beta>0$.

For $0<r\leq \varepsilon$ we get by a similar argument
\begin{align*}
&\mathcal{L}_\varepsilon V_{0}(r,y)
\\
&\leq-\beta r_{0}e^{-\varepsilon}
+\sup_{0<y<\infty,0<r \leq \varepsilon}
\left\{(-\sigma^{2}r^{2} \varepsilon^{2\gamma-2} +2\beta y)e^{-y}+
 \left(\beta r-y+\frac{1}{2}\sigma^{2}r^{2} \varepsilon^{2\gamma-2} \right)e^{-r}\right\}
\\
&\leq-\beta r_{0}e^{-\varepsilon}
+\sup_{0<y<\infty}2\beta ye^{-y}+
 \sup_{0 <r \leq \varepsilon}
 \left(\beta r+\frac{1}{2}\sigma^{2}r^{2} \varepsilon^{2\gamma-2} \right)e^{-r}
\\
&\leq-\beta r_{0}e^{-\varepsilon}
+\frac{2\beta}{e}+\left(\beta \varepsilon + \frac12 \sigma^{2} \varepsilon^{2\gamma}\right)<0,
\end{align*}
for any sufficiently large $r_{0}$ such that
\begin{equation}\label{r0epsilon}
r_{0}>\frac{e^{\varepsilon}}{\beta}
\left[\frac{2\beta}{e}+\left(\beta \varepsilon+ \frac12 \sigma^{2}\varepsilon^{2\gamma}\right)\right].
\end{equation}
Since $R\geq\varepsilon$, the condition \eqref{r0R} implies \eqref{r0epsilon}.
Thus $\mathcal{L}_{\varepsilon}V_{0}(r,y)$ 
for $0<r<2R$ as long as \eqref{r0R} holds.

(b) If $r\geq 2R$, then we must have $0<y<2R$.
Then, we have
\begin{align*}
\mathcal{L}_\varepsilon V_{0}(r,y)
&=(-\sigma^{2}r^{2\gamma}+2\beta y)e^{-y}
 +\left(\beta r-y-\beta r_{0}+\frac{1}{2}\sigma^{2}r^{2\gamma}\right)e^{-r}
\\
&\leq (-\sigma^{2}r^{2\gamma}+4\beta R)e^{-y}+
   \left(\beta r+\frac{1}{2}\sigma^{2}r^{2\gamma}\right)e^{-r}
-\beta r_{0}e^{-r}.
\end{align*}

For $\frac{1}{2}<\gamma\leq 1$ we have
\begin{equation}\label{mid}
\mathcal{L}_\varepsilon V_{0}(r,y) 
\leq -\sigma^{2}r^{2\gamma}e^{-2R} + 4\beta R +
\left(\beta r+\frac{1}{2}\sigma^{2}(r^2+r)\right)e^{-r}
-\beta r_{0}e^{-r},
\end{equation}
where we used the fact that $r^{2\gamma}\leq r^{2}+r$ for every $r\geq 0$ and $\frac{1}{2}<\gamma\leq 1$.
Moreover, for any $r\geq 0$, we have $e^{r}\geq r+\frac{1}{2}r^{2}$ so that
\begin{equation}\label{pluginto}
\left(\beta r+\frac{1}{2}\sigma^{2}(r^2+r)\right)e^{-r}
\leq\frac{(\beta+\frac{1}{2}\sigma^{2})r+\frac{1}{2}\sigma^{2}r^{2}}{r+\frac{1}{2}r^{2}}
\leq\beta+\sigma^{2}.
\end{equation}
Hence, by plugging \eqref{pluginto} into \eqref{mid}, we get
\begin{equation}
\mathcal{L}_\varepsilon V_{0}(r,y) 
\leq -\sigma^{2}r^{2\gamma}e^{-2R} + 4\beta R +
\beta+\sigma^{2}
-\beta r_{0}e^{-r}.
\end{equation}
Denote $H(r):=\sigma^{2}r^{2\gamma}e^{-2R}+\beta r_{0}e^{-r}$, $r\geq 0$.
Let us give a lower bound of $H(r)$ over $r\geq 0$.
For $0\leq r\leq 1$, we have $H(r)\geq\frac{\beta}{e}r_{0}$,
and for $r>1$, we have $H(r)\geq\sigma^{2}re^{-2R}+\beta r_{0}e^{-r}$ since $\gamma\in(\frac{1}{2},1]$.
Denote $\tilde{H}(r):=\sigma^{2}re^{-2R}+\beta r_{0}e^{-r}$, and we can compute that
\begin{equation}
\tilde{H}'(r)=\sigma^{2}e^{-2R}-\beta r_{0}e^{-r},
\end{equation}
which is negative for $r<\log(\frac{\beta r_{0}}{\sigma^{2}})+2R$
and positive for $r>\log(\frac{\beta r_{0}}{\sigma^{2}})+2R$. 
Thus
\begin{equation}
\tilde{H}(r)\geq\sigma^{2}e^{-2R}\left[\log\left(\frac{\beta r_{0}}{\sigma^{2}}\right)+2R\right]
+\beta r_{0}e^{-\log(\frac{\beta r_{0}}{\sigma^{2}})-2R}
=\sigma^{2}e^{-2R}\left[\log\left(\frac{\beta r_{0}}{\sigma^{2}}\right)+2R+1\right].
\end{equation}
Hence, 
\begin{equation}
H(r)=\sigma^{2}r^{2\gamma}e^{-2R}+\beta r_{0}e^{-r}
\geq
\min\left\{\frac{\beta}{e}r_{0},
\sigma^{2}e^{-2R}\left[\log\left(\frac{\beta r_{0}}{\sigma^{2}}\right)+2R+1\right]\right\}.
\end{equation}
Hence, we conclude that
\begin{equation}
\max_{r,y\geq 0}\mathcal{L}_{\varepsilon}V_{0}(r,y)<0,
\end{equation}
if $r_{0}$ is sufficiently large so that
\begin{equation}
\min\left\{\frac{\beta}{e}r_{0},
\sigma^{2}e^{-2R}\left[\log\left(\frac{\beta r_{0}}{\sigma^{2}}\right)+2R+1\right]\right\}
>4\beta R+\beta+\sigma^{2},
\end{equation}
which holds if
\begin{equation}
r_{0}>\max\left\{\frac{e}{\beta}(4\beta R+\beta+\sigma^{2}),\frac{\sigma^{2}}{\beta}
e^{\frac{e^{2R}}{\sigma^{2}}(4\beta R+\beta+\sigma^{2})-2R-1}\right\}.
\end{equation}

For both cases (a) and (b), for sufficiently large $r_0$ the inequality
$\max_{r,y\geq 0}\mathcal{L}_\varepsilon V_{0}(r,y)<0$ is satisfied.  
The proof is complete.
\end{proof}


\section*{Acknowledgements}

We would like to thank Camelia Pop for discussions about boundary conditions of SDEs. 
The authors are also grateful to the REU students Ruby Oates, Alex Pollack and Kelsey Paetschow
for their help with Figure~\ref{fig:beta}.
Lingjiong Zhu acknowledges the support from NSF Grant DMS-1613164.

\end{document}